\documentclass[conference]{IEEEtran}
\usepackage{cite}
\usepackage[pdftex]{graphicx}
\usepackage{amsmath}
\usepackage{amssymb}
\usepackage{amsthm}
\usepackage{hyperref}
\usepackage{listings}
\theoremstyle{plain}

\newtheorem{corollary}{Corollary}
\newtheorem{lemma}{Lemma}
\theoremstyle{definition}
\newtheorem{definition}{Definition}

\newtheorem{example}{Example}

\makeatletter
\def\ps@headings{%
\def\@oddhead{\mbox{}\scriptsize\rightmark \hfil \thepage}%
\def\@evenhead{\scriptsize\thepage \hfil \leftmark\mbox{}}%
\def\@oddfoot{\scriptsize \@date\hfil Approved for public release; unlimited distribution. Not export controlled per ES-FL-011422-0005.}%
\def\@evenfoot{\scriptsize Approved for public release; unlimited distribution. Not export controlled per ES-FL-011422-0005.\hfil \@date}}

\def\ps@IEEEtitlepagestyle{%
\def\@oddhead{\mbox{}\scriptsize\rightmark \hfil \thepage}%
\def\@evenhead{\scriptsize\thepage \hfil \leftmark\mbox{}}%
\def\@oddfoot{Approved for public release; unlimited distribution. Not export controlled per ES-FL-011422-0005.\hfil}%
\def\@evenfoot{Approved for public release; unlimited distribution. Not export controlled per ES-FL-011422-0005.\hfil}}

\makeatother

\begin{document}

\title{Statistical detection of format dialects using the weighted Dowker complex}

\author{
  \IEEEauthorblockN{Michael Robinson}
  \IEEEauthorblockA{Department of Mathematics and Statistics\\
    American University\\
    Washington, DC\\
    Email: michaelr@american.edu}\\
  
  \IEEEauthorblockN{Letitia W. Li}
  \IEEEauthorblockA{BAE Systems FAST Labs\\
    Arlington, VA\\
    Email: letitia.li@baesystems.com}\\
\and
  \IEEEauthorblockN{Cory Anderson}
  \IEEEauthorblockA{BAE Systems FAST Labs\\
    Arlington, VA\\
    Email: cory.s.anderson@baesystems.com}\\
  
  \IEEEauthorblockN{Steve Huntsman}
  \IEEEauthorblockA{Arlington, VA\\
    Email: sch213@nyu.edu}
}

\maketitle

\begin{abstract}
  This paper provides an experimentally validated, probabilistic model of file behavior when consumed by a set of pre-existing parsers.
  File behavior is measured by way of a standardized set of Boolean ``messages'' produced as the files are read.
By thresholding the posterior probability that a file exhibiting a particular set of messages is from a particular dialect, our model yields a practical classification algorithm for two dialects.
We demonstrate that this thresholding algorithm for two dialects can be bootstrapped from a training set consisting primarily of one dialect.
Both the (parametric) theoretical and the (non-parametric) empirical distributions of file behaviors for one dialect yield good classification performance, and outperform classification based on simply counting messages.

Our theoretical framework relies on statistical independence of messages within each dialect.
Violations of this assumption are detectable and allow a format analyst to identify ``boundaries'' between dialects.
A format analyst can therefore greatly reduce the number of files they need to consider when crafting new criteria for dialect detection, since they need only consider the files that exhibit ambiguous message patterns.
\end{abstract}


\begin{center}\end{center}

% \fi
\IEEEpeerreviewmaketitle

\section{Introduction}

This paper provides an experimentally validated, probabilistic model of file behavior on a set of Boolean features (``messages'').
By thresholding the posterior probability that a file exhibiting a particular set of messages is from a particular dialect, our model yields a practical classification algorithm for two dialects.
We demonstrate that this thresholding algorithm for two dialects can be bootstrapped from a training set consisting primarily of one dialect.
Both the (parametric) theoretical and the (non-parametric) empirical distributions of file behaviors for one dialect yield good classification performance.  Furthermore, although message count may correlate with one dialect over another, we show that our approach yields a substantially better classifier because it correctly dispositions the files that do not follow this trend.

Although our theoretical framework relies on statistical independence of messages within each dialect, violations of this assumption are detectable and allow a format analyst to identify ``boundaries'' between dialects.
The key payoff of this approach is that we can tell which message patterns are easy to disposition as one dialect or another, and we can also identify which message patterns are ambiguous.  This information allows an analyst to pinpoint specific features that account for the files on which two (or more) parsers disagree.

\subsection{Application context}

File format specifications are dynamic entities, and are often ambiguous.
A given clause in a specification may have several distinct but self-consistent interpretations,
and these interpretations may impact the interpretations of other, related clauses.
As a result, files from different dialects of a format tend to exhibit divergent behavior at multiple,
independent points within a parser's code base.
Our methodology exploits this independence structure to discriminate between dialects.

When standards organizations attempt to resolve an ambiguity in the specification,
stakeholders bring example files that exhibit specific behaviors.
One may suspect that these files are not unbiased samples from a statistical perspective!
Presently, there appears to be no unbiased way to query the corpus of files ``in the wild'' to find examples of files whose behavior is ambiguous.
Our methodology provides an even-handed, systematic way to use an existing set of parsers and a large corpus of files to identify specific parser message patterns that are either easy (or difficult) to disposition as one dialect or another.
Filtering for these message patterns can enable an expert user to identify a sample of files that are representative of the dominant dialects with different interpretations of the specification.

\subsection{Background}

Beyond what our team has previously published in the past year or so, there appears to be very little work in analyzing file behavior using statistical tools \cite{Robinson_looking_2021,Ambrose_2020}.  In contrast, nearly all existing file format analysis  uses the structure of file \emph{contents} rather than the responses of parsers to those contents (for instance, see \cite{belaoued2015real,al2018ransomware, 8685181,ALAZAB201591,demme2013feasibility}).

Most relevant to our work, \cite{burgcsvfiles} notes that $39$ valid dialects for CSV exist, which makes parsing any given CSV file challenging. Dialects all contain their own set of quotation marks, escape characters, delimiters, headers, comment text, etc. Their approach identifies the dialect of a file by using consistency measures that are expected to score higher if the dialect is identified. When the dialect and its associated delimiters are used correctly in a parse, the number of cells in a row should all be the same, and the types of cells in a single column should all be identical. The PADS project aids users by generating a formal language description of an \emph{ad hoc} format, including a inference algorithm that uses statistical histograms of tokens to identify if tokens should be combined into arrays or structures \cite{fisher2008dirt}.

Statistical analysis of files has also been used for malicious file detection or analysis of collected drives as digital forensic evidence. The DIRIM tool detects suspicious files or drives based on file metadata \cite{rowe2011finding}. It uses PCA and $k$-means to cluster files and determine features more likely indicative of files of interest, such as those attempting to conceal their file type with a misleading file extension.  Statistical features based upon file actions has also been used to identify certain malicious behaviors \cite{Scofield_2017}.  Although they start with data that are formatted similarly to ours, their ultimate goal is simply classification rather than dialect identification.

In other applications, clustering has been used to identify natural language dialects. Lundberg showed that one can relate clusters in recordings from Swedish speakers to spoken dialects. Recordings were converted into acoustic features, and clustered using PCA, $k$-means, and hierarchical clustering \cite{lundberg2005classifying}.  Grieve \emph{et al.} also used hierarchical clustering for classification of dialects across different regions of the US, based on use of lexical alternation variables (for instance, ``actually'' versus ``in fact'') \cite{grieve2011statistical}.

Finally, we note in passing the structural similarity between the weighted Dowker complex approach used here and \emph{factor analysis} \cite{yong2013beginner}.  Factor analysis effectively works from the opposite perspective to ours.  One starts with an assumed dependence between variables rather than locating violations of independence when they occur in the data.  The difference is important; violations of independence occur near ``dialect boundaries'' within a dataset.

\subsection{Assumptions and Limitations}

This paper assumes that there is a pre-existing collection of Boolean ``messages'' produced by several parsers, and that each file under consideration has been processed by each parser.  We will assume that these messages cover all of the relevant aspects of the dialects that we wish to measure, and that the messages are diverse enough to discriminate between these aspects.  Our methodology is format agnostic, in that it does not look at the file contents directly.  File contents are only considered through the lens of the pre-existing parsers, so a user of our method will not need to be a format expert.  Moreover, ``messages'' need not be error messages, \emph{per se}.  For instance, a message could be the presence or absence of a certain byte sequence in the file, or it may simply report whether the exit code for a given parser corresponds to a valid parse.

The theoretical justification for our method relies upon the independence of messages for files within a given dialect.  In a representative sample of messages for the dataset we describe in the next section, we found that most pairs of messages are independent, though a small subset of messages are highly dependent on other messages (see Section \ref{sec:independence_check}).

In our specific dataset, this dependence arises because several parsers can be run with different options.  Running the same parser with different options sometimes results in nearly identical messages being produced.

The proper strategy from the perspective of a formal model would be to capture the dependence structure, but it is \emph{extremely} computationally infeasible to test for dependence (even pairwise dependence) whenever there are more than a few dozen messages.  When we incorrectly assume conditional independence across all messages---as we will blithely proceed to do---this artificially reduces the probability of certain patterns of messages below what is actually observed.  We can compensate for this effect by raising the overall message probability above what is estimated on a per-message basis.  This appears to result in a distribution of message patterns that agrees with the observations, though some of the multi-way dependence structure is lost.

Finally, we will assume that for each dialect that we wish to study, there are only two kinds of messages: those that occur frequently for that dialect, and those that occur at about the same frequency as for other dialects.  Although apparently limiting, our dataset agrees with this assumption (see Figure \ref{fig:goodA_mp}).  The less frequent kind of message is effectively a ``background'' message.  At the start of our analysis, we do not know which message plays which role for any given dialect.  

%%%%%%%%%%%%%%%%%%%%%%%%%%%%%%%%%%%%%%%%%%%%%%%%

\section{Dataset description}

The data we processed to test our methodology were produced as training data by the Test and Evaluation Team for the DARPA SafeDocs evaluation exercise 3.  These data consist of PDF files, ostensibly compliant with the ISO 32000-2 standard.
For this exercise, we used the ``Universe A'' \emph{good files} and \emph{bad files} datasets.  Each dataset consisted of $100001$ hand-curated PDF files, for a total of $200002$ files.

Given that the Test and Evaluation Team consists of PDF format experts, the Test and Evaluation Team was able to manually ensure that these two datasets had known ground truth: the good files are either syntactically and semantically valid PDF files, or are files that could be unambiguously corrected.  The bad files exhibit various kinds of malformations including syntax errors, semantic violations, or other kinds of problems.
The good files were largely sourced from Common Crawl \cite{common_crawl}, 
while the bad files were drawn from various sources: some were found in the wild (from Common Crawl), some were
malicious files created by the Test and Evaluation Team, and others were non-malicious non-compliant files created by the Test and Evaluation Team.

Each file was processed through $13$ distinct base parsers, run with various options to make a total of $29$ parsers.
A total of $\#M = 3022$ Boolean messages were collected, as shown in Table \ref{tab:parsers}.  One message per parser is an exit code corresponding to the presence of an error, which accounts for a total of $29$ messages.  The rest of the messages correspond to specific regular expressions (regexes) run against {\tt stderr}, as explained in Section \ref{sec:message_regex_construction}.  Several of these messages were found to play an important role in identifying dialects and are discussed in detail in latter sections of this paper.  The reader interested in seeing example regexes should consult Table \ref{tab:weighty_messages}.

\begin{table}
    \begin{center}
      \caption{Parsers and options used}
      \label{tab:parsers}
    \begin{tabular}{|l|l|c|}
      \hline
      Parser& Possible options&Messages\\
      \hline
      \hline
{\tt caradoc}&{\tt extract}&121\\
&{\tt stats}&121\\
&{\tt stats --strict}&94\\
\hline
{\tt hammer}& (none)&69\\
\hline
{\tt mutool}&{\tt clean}&214\\
&{\tt draw}&248\\
&{\tt show}&75\\
\hline
{\tt origami}&{\tt pdfcop}&40\\
\hline
{\tt pdfium}& (none)&26\\
\hline
{\tt pdfminer}&{\tt dumppdf}&88\\
&{\tt pdf2txt}&155\\
\hline
{\tt pdftk}&{\tt server}&33\\
\hline
{\tt pdftools}&{\tt pdfid}&4\\
&{\tt pdfparser}&30\\
\hline
{\tt peepdf}& (none)&4\\
\hline
{\tt poppler}&{\tt pdffonts}&100\\
&{\tt pdfinfo}&90\\
&{\tt pdftocairo}&214\\
&{\tt pdftoppm}&155\\
&{\tt pdftops}&189\\
&{\tt pdftotext}&139\\
\hline
{\tt qpdf}& (none)&192\\
\hline
{\tt verapdf}&{\tt greenfield}&40\\
&{\tt pdfbox}&50\\
\hline
{\tt xpdf}&{\tt pdffonts}&82\\
&{\tt pdfinfo}&70\\
&{\tt pdftoppm}&122\\
&{\tt pdftops}&157\\
&{\tt pdftotext}&100\\
\hline
\hline
Total&&3022\\
\hline
    \end{tabular}
  \end{center}
\end{table}

The input to the methodology described in this paper is therefore an unordered list of file-message pairs, recording the set of messages that occurred for each file.  These data can be rendered into a matrix form, in which the rows correspond to messages and the columns correspond to files.  The entries are either $1$ or $0$, if the message occurred or did not occur, respectively.  The matrices for both datasets are shown in Figure \ref{fig:uniA_matrices}.  Since the same set of messages was collected for both datasets, the rows have the same meaning in both matrices.  Even though the same number of files was present in each dataset, the meaning of the columns differs since the sets of files differ.

\begin{figure*}
  \begin{center}
    \includegraphics[width=5in]{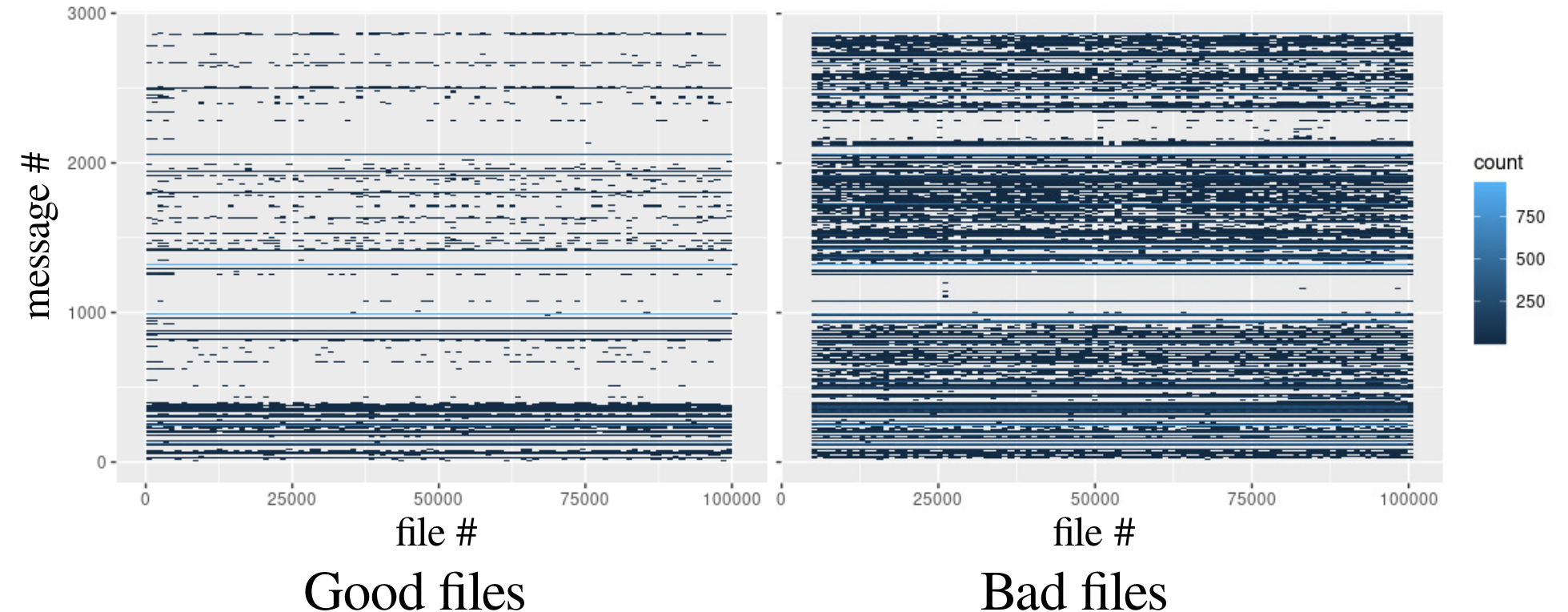}
    \caption{Matrices of messages (rows) versus files (columns) for the SafeDocs Evaluation 3 Universe A good files (left) and bad files (right).  The message numbers are arbitrary but are the same for both matrices.  The files differ between the two data sets.}
    \label{fig:uniA_matrices}
  \end{center}
\end{figure*}

It is immediately clear visually that the two matrices are quite different.  In particular, the bad files produce far more messages than the good files on average.  A rough classification of an unknown file from one of these two sets based on its message count can clearly be an effective strategy.  However, as will be shown in Section \ref{sec:posterior_performance}, our method outperforms this na\"ive strategy by a wide margin.  One can determine which specific message patterns correspond to different behaviors, yielding a finer classification.

\subsection{Message regex construction}
\label{sec:message_regex_construction}
The message regexes were generated by running all unique {\tt stderr} messages for each parser (independently in parallel) through a set of manually created find-and-replace rules, followed by a final multi/single line filter, as described below.  The idea is that message regexes should match to a message type or template, ignoring variable fields. For example, one of the caradoc stderr messages is {\tt PDF error : Error in Flate/Zlib stream in object }\emph{[number]}{\tt !}, and messages which fill in the template with different object numbers should all still be combined into the same regex: {\tt PDF error : Error in Flate/Zlib stream in object {\textbackslash}d+ !}

Each rule includes (1) a {\tt stderr}-file-wide regex, (2) a find regex, and (3) replace text.
If the rule's {\tt stderr}-file-wide regex matches the {\tt stderr} file produced by a parser,
then the rule's find regex is used to replace all \emph{its} matches with the rule's replace text to create a potential message regex.

Newlines are not automatically treated as a message delimiter, so the resulting regex set contains duplicates of single/multi line messages in different orders.
To compress these duplicates, a final multi/single line filter numbers all unique lines and repeatedly loops through the sequences to distinguish single line regexes from multi-line regexes.  To this end, it uses the known single lines to identify and split more single lines, until a final set of single lines and unsplittable multi-line message regexes is obtained for each parser.

\subsection{Preliminary test: Message independence}
\label{sec:independence_check}
To assess the message conditional independence assumption, we ran $\chi^2$ tests for independence on each pair of messages drawn from a simple random sample of $n=30$ messages for the good files and the bad files, separately.  Figure \ref{fig:uniA_chisq} shows a summary of the resulting $p$-values for each pair of messages from the sample for both the good (left) and bad (right) files.  A small $p$-value corresponds to a likely dependent pair of messages, while a large $p$-value corresponds to a pair of messages that are likely independent.  Figure \ref{fig:uniA_chisq} shows that most message pairs are likely independent (with $p>0.05$, though typically much larger).  There is a small subset of the message pairs that are very dependent (with $p<0.05$).  Given the banding structure, these dependencies are caused by a small number of individual messages.

\begin{figure*}
  \begin{center}
    \includegraphics[width=5in]{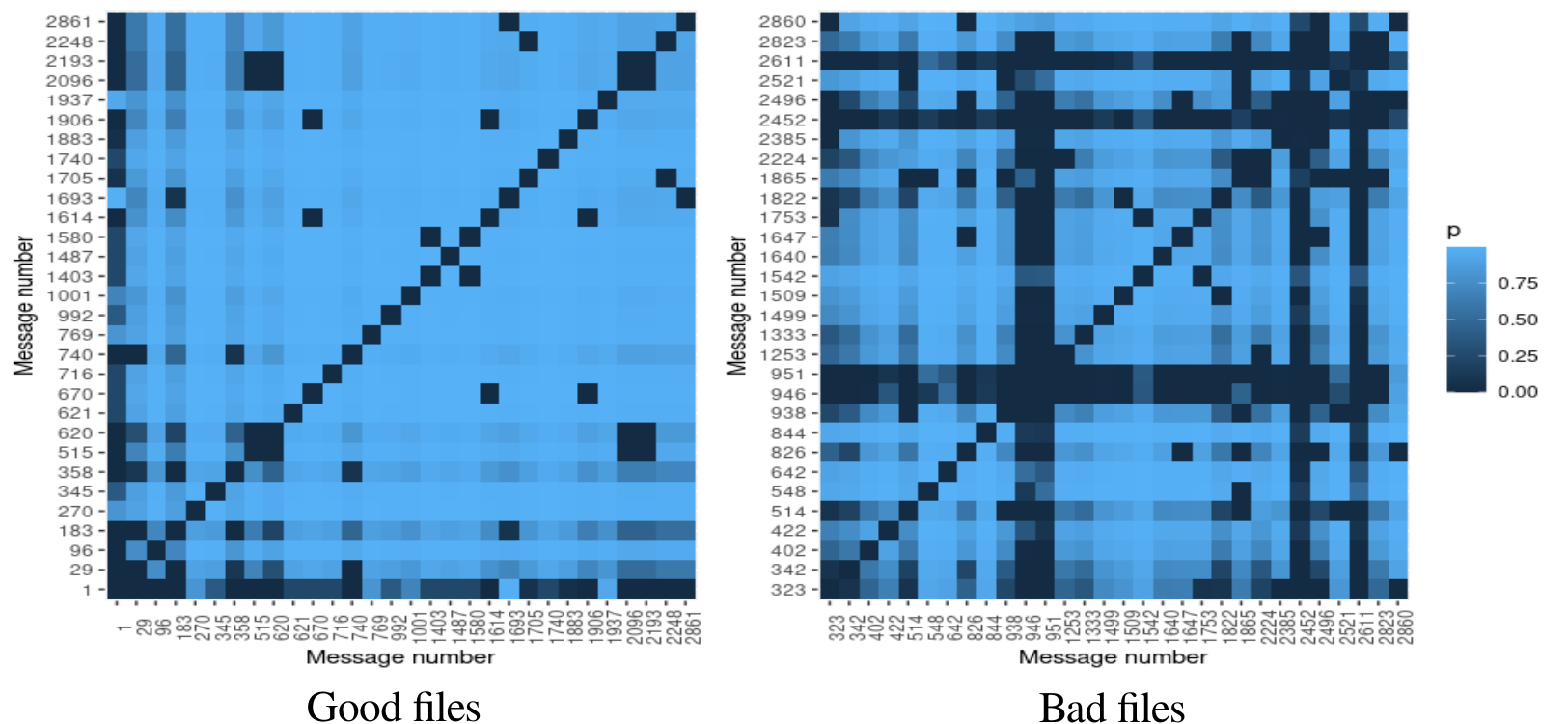}
    \caption{SafeDocs Evaluation 3 Universe A dataset pairwise $\chi^2$ test for independence between a random sample of $30$ messages. The $p$-values for each pair are shown for the good files (left) and bad files (right).}
    \label{fig:uniA_chisq}
  \end{center}
\end{figure*}

Some of the most significant message dependencies correspond to duplicate regexes run on the output of a given parser, when different options are enabled.  For instance, messages $19$ and $140$ in Table \ref{tab:weighty_messages} use the same regex for the output of the {\tt caradoc} parser.  Message $19$ is reported when {\tt caradoc} is run with the {\tt extract} option, while message $140$ is reported when {\tt caradoc} is run with the {\tt stats} option.  Given that running {\tt caradoc} with different options likely uses the same code in many places, is reasonable to expect---though it need not be the case that---a given file will produce both of these messages or neither of them.  

%%%%%%%%%%%%%%%%%%%%%%%%%%%%%%%%%%%%%%%%%%%%%%%%

\section{Methods}

Let $A$ and $B$ be two sets of files, corresponding to different dialects that we would like to classify.
That is, files in $A$ are of one dialect, while files in $B$ are of another dialect.
Each of these files are run through parsers that can potentially produce any messages from a fixed set $M$ of messages.
We will assume that the messages are independent as random variables, after they have been conditioned upon the dialect.
That is, if we are only considering files of dialect $A$, then the messages will be independent.
However, if we consider the files in two dialects $A \cup B$, then independence may be violated.

The independence assumption lets us consider the message probabilities for each dialect separately.
In the case of dialect $A$, we will model the messages as Bernoulli random variables with one of two probabilities: $p_0$ or $p_A$.
The subset of messages with probability $p_A$ is denoted as $M_A \subseteq M$.
The remainder of the messages (in $M_A^c = M - M_A$, the complement of $M_A$ within $M$) are assumed to occur with probability $p_0$.
If we assume that $p_A > p_0$, it is useful to interpret $M_A$ as the set of messages that are ``characteristic'' to dialect $A$.
We can similarly define a set $M_B$ of messages that occur with probability $p_B$ for files in dialect $B$.

We will use only one value for $p_0$ across both dialects, which suggests the interpretation that $p_0$ is the ``background'' message probability.
If message is in $M_A \cap M_B^c$, then it will either occur with probability $p_A$ if the file is in dialect $A$, or it will occur with probability $p_0$ if the file is in dialect $B$.
Conversely, a message in $M_A^c \cap M_B$ will occur with probability $p_B$ if the file is in dialect $B$, or it will occur with probability $p_0$ if the file is in dialect $A$.

Under the above assumptions, the probability of getting exactly a set of messages $K \subseteq M$ on a file $f$ in dialect $A$ is
\begin{equation}
  \label{eq:message_prob}
  \begin{aligned}
    P(K | A) =& p_0^{\#(K \cap M_A^c)} (1-p_0)^{\#(K^c \cap M_A^c)} \times    \\& p_A^{\#(K\cap M_A)} (1-p_A)^{\#(K^c\cap M_A)}.
    \end{aligned}
\end{equation}

We note that in our two datasets, the messages are not completely independent, especially for those messages that have a low probability of occurrence.  The proper probabilistic model should account for various correlations between messages, but is substantially more complicated than Equation \eqref{eq:message_prob}.  That said, by artificially inflating the probability $p_0$, one can produce similar message patterns to what are observed in the data.

\subsection{Estimating the parameters of the model for a single dialect}
\label{sec:estimation}

If $\#M_A \ll \#M$, then we can estimate $p_0$ from the data.  The probability that no messages will occur is then
\begin{equation}
    \label{eq:p0_estimator}
\begin{aligned}
  P(\emptyset | A) &= (1-p_0)^{\#M_A^c} (1-p_A)^{\#M_A}\\
  &= (1-p_0)^{\#M} \left(\frac{1-p_A}{1-p_0}\right)^{\#M_A}\\
  &\approx (1-p_0)^{\#M}.
\end{aligned}
\end{equation}

On the other hand if $p_A \gg p_0$ is a good assumption, then considering each message independently will identify those that are in $M_A$, since their individual probabilities differ from $p_0$ by a significant amount.

Although it turns out to be unnecessary in the case of our dataset, one could identify messages in $M_A$ by a standard hypothesis test for a proportion.  This is helpful if $p_A$ is close to $p_0$.  To that end, if the $t$-test statistic for message $k$,
\begin{equation*}
  t=\frac{p_k- p_0}{\sqrt{\frac{p_0(1-p_0)}{\#\text{files}}}}
\end{equation*}
is large---say larger than 1.96 for 95\% confidence---then we conclude that this message is an element of $M_A$.

%The following formula is for \emph{exactly} one message $m$ happening:
%\begin{equation*}
%  P(\{m\} | A) = \begin{cases}
%    \text{if } m\in M_A & (1-p_0)^{\#M-\#M_A} p_A (1-p_A)^{\#M_A-1}\\
%    \text{if } m\notin M_A & p_0(1-p_0)^{\#M-\#M_A-1} (1-p_A)^{\#M_A}\\
%    \end{cases}
%\end{equation*}
%thus the messages will have two distinct probabilities of occurance, depending on whether or not they are in $M_A$.
%The case of exactly one message probably does not typically happen very often, since the probabilities will be quite low.

\subsection{Message patterns and weighted Dowker complexes}
\label{sec:mp_and_dowker}

Given the set of messages $M$, there are $2^M$ possible \emph{message patterns} that might occur for a given a file.  Under the model given by Equation \eqref{eq:message_prob}, not all of these are equally likely.  Some message patterns can be expected to be quite common.  For instance, if the messages are all ``errors'' and the dialect under consideration consists of mostly valid files, we should expect the the empty pattern $\emptyset$ to be the most common.

The set of message patterns that occur for a given dataset has rich mathematical structure.  The most famous of these structures is that of the \emph{Dowker complex}.

\begin{definition} \cite{Robinson_2021,Ambrose_2020}
  The set $X$ of all message patterns $K \subseteq 2^M$ such that there is a file exhibiting (at least) the messages in $K$ is called the \emph{Dowker complex}.  Each such message pattern is called a \emph{Dower simplex}.  Furthermore, the number of files exhibiting \emph{exactly} the messages in $K$ \emph{and no others} is called the \emph{differential weight} $d(K)$.  For simplicity, we will usually call $d(K)$ the \emph{weight} or the \emph{file count} for the message pattern $K$.
\end{definition}

There are many interesting properties of the Dowker complex, because it is an example of an \emph{abstract simplicial complex}, a combinatorial topological model of an abstract space.  For our purposes in this article, the most important properties follow from the fact that message patterns are \emph{partially ordered} by \emph{subset inclusion}.  That is, if $K_1$, $K_2$, and $K_3$ are message patterns and we know that $K_1 \subseteq K_2$ and $K_2 \subseteq K_3$, then it follows that $K_1 \subseteq K_3$.  Moreover, if $K_1 \subseteq K_2$ and $K_2 \subseteq K_1$, then it follows that $K_1 = K_2$.  It is obvious that the message count for a message pattern $K$ (the number of messages in $K$) constrains which other message patterns are related to $K$.  This impacts the statistics of the distribution of message counts exhibited within a given dataset, as will be explained in Lemma \ref{lem:mc_sorting} of Section \ref{sec:message_counts}.

Since the number of messages is typically large, for instance $\# M = 3022$, the number of possible message patterns is truly enormous.  It is therefore unwise to attempt to compute the weight of all possible message patterns for a given dataset.  Since most of these weights will be zero, it is much better to compute the message patterns that are actually present and their corresponding weights simultaneously.  This can be done efficiently by lexically sorting the columns of the matrix form of the data, and then grouping blocks of identical columns greedily.  Each distinct column clearly corresponds to a particular message pattern with a nonzero weight, and the weight is simply the number of duplicate columns.

As a sample implementation of this greedy approach, we exhibit a very succinct implementation in the {\tt tidyverse} dialect of the R statistical programming language.  (A more optimized Python implementation, suitable for interactive exploration of our entire test data, is discussed in Section \ref{sec:dowker_display}.) This implementation assumes that the Boolean matrix of messages is stored as a data frame called {\tt message\_data}, in which the rows correspond to files.  (That is, the data frame is the transpose of the matrices shown in Figure \ref{fig:uniA_matrices}.)  The columns which correspond to messages are named starting with the character {\tt 'X'}, and all other columns are ignored.

The following snippet will compute a new data frame called {\tt dowker} containing the Dowker complex and weights (file counts), and then produces a histogram of the weights, like what is shown in Figure \ref{fig:mc_vs_weight}. Because the {\tt count()} function is quite efficient, the snippet runs quickly provided that the data fit in memory.

\begin{lstlisting}[language=R]
dowker <- message_data %>%
  group_by(across(starts_with('X'))) %>%
  count(name='weight',sort=TRUE) %>%
  ungroup()
dowker %>% 
  mutate(simplex=row_number()) %>%  
  ggplot(aes(x=simplex,
             y=sort(weight,
                    decreasing=TRUE))) + 
  geom_line()
\end{lstlisting}

\subsection{Relationship with message count}
\label{sec:message_counts}

If the messages mostly correspond to error conditions and most files in a given dialect are valid, then we expect that most files will generate few messages.  Using the model given by Equation \eqref{eq:message_prob}, the probability that a file will produce $n$ messages is the weighted sum of binomial distributions,
\begin{equation*}
  \begin{aligned}
    P(\#K = n | A) = \sum_{k = 0}^{n} & \left[\binom{\#M_A}{k}\binom{\#M - \#M_A}{n-k}p_0^{n-k} \times \right.\\ & (1-p_0)^{\#M-\#M_A-(n-k)} \times \\ &\left.p_A^{k}(1-p_A)^{\#M_A-k}\right]
    \end{aligned}
\end{equation*}

\begin{figure}
  \begin{center}
    \includegraphics[width=3in]{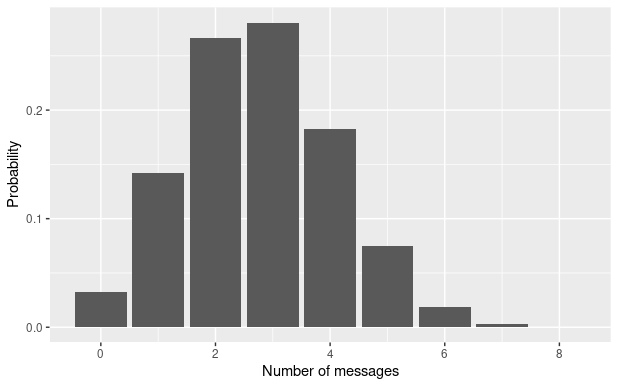}
    \caption{Probability that a file will produce a certain number of messages, given a total of $8$ messages: $3$ messages with probability $0.25$, and $5$ messages with probability $0.4$.}
    \label{fig:theoretical_mc}
  \end{center}
\end{figure}

After a bit of algebra (or logic), if $p_A=p_0$, then the probability of $n$ messages occurring is simply given by a binomial distribution.
Since the binomial distribution is not necessarily monotonic, neither is message count, for instance see Figure \ref{fig:theoretical_mc}.
Nevertheless, there is a definite dependence between message count and the probability of a given message pattern.

\begin{lemma}
  \label{lem:mc_sorting}
  Assume that each message has probability less than $0.5$ and the messages are independent when conditioned on files of dialect $A$.  If $K_1 \subset K_2$ are two message sets, so that $\#K_1 < \#K_2$, then $P(K_2 | A) < P(K_1|A)$.
\end{lemma}

\begin{proof}
  By moving from $K_1$ to $K_2$, we are merely swapping out factors in $P(K|A)$ of the form $(1-p_k)$ for corresponding factors $p_k$.  Under the hypothesis, these new factors are smaller.  
\end{proof}

If messages occur with probability greater than $0.5$, it is usually more informative to consider their \emph{absence} instead of their presence.  For instance, {\tt pdfium} usually produced a message \verb;Processed \d+ pages\.;.  The absence of this message suggests that the parser exited without producing any useful output because the file was too malformed, but its presence was not very helpful.

\begin{figure}
  \begin{center}
    \includegraphics[width=3.25in]{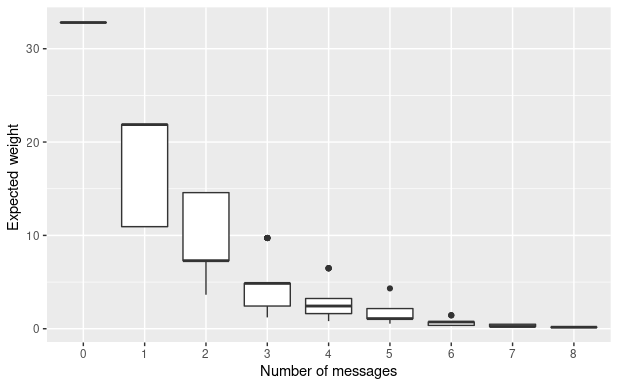}
    \caption{Expected weight (file count) of a simplex compared with its message count, given a total of $8$ messages: $3$ messages with probability $0.25$, and $5$ messages with probability $0.4$.}
    \label{fig:mc_vs_weight}
  \end{center}
\end{figure}

Note that the Lemma is clearly true for Equation \eqref{eq:message_prob}, though it holds even if every message probability is different.
Inspired by the binomial distribution, we expect the weight to eventually decrease as number of messages increases, as shown in Figure \ref{fig:mc_vs_weight}.  The variability in weights in Figure \ref{fig:mc_vs_weight} is a bit misleading, because as messages are added the weight must decrease.

\begin{corollary}
  \label{cor:inconsistent_edges}
  Under the same conditions as Lemma \ref{lem:mc_sorting}, if $K_1$ and $K_2$ are two Dowker simplices, with $K_1 \subset K_2$ (so that $K_2$ consists of more messages happening), then the expected weights satisfy $d(K_1 ) > d(K_2)$.
\end{corollary}

That is, patterns with a higher message count are typically exhibited by fewer files.  Note that this does not mean that the weight decreases with increasing message count; only that it decreases \emph{as additional messages are considered}.  As a result, it is particularly interesting when this trend is not followed in the actual data: the weight \emph{increases} as additional messages are added.  These violations of Corollary \ref{cor:inconsistent_edges} typically occur when messages are strongly dependent upon one another, and are effectively measuring related phenomena.

\subsection{Distribution of weights}

Observe that for a given pattern of messages $K$ for a single dialect $A$, Equation \eqref{eq:message_prob} gives the probability that a particular file will count towards the weight for $K$, when $K$ is thought of as a simplex of the Dowker complex.

A convenient display of the weights for a given set of message patterns is to sort them in decreasing order, resulting in a \emph{Dowker histogram}.  Equation \eqref{eq:message_prob} specifies the expected values of each possible weight, but does not specify how many simplices will have this particular weight.  This is easily found, however.  There are
\begin{equation*}
  \binom{\# M_A}{k} \binom{\#M - \# M_A}{n-k}
\end{equation*}
different simplices corresponding to the situation where $n$ messages occurred, $k$ of which are characteristic to the dialect $A$.  Each of these simplices has the same expected weight, namely
\begin{equation}
  \label{eq:dowker_histogram}
  \begin{aligned}
  P(\#K &= n, \#(K\cap M_A) = k | A) = \\
 & p_0^{n-k} (1-p_0)^{(\#M - \# M_A - (n-k))} p_A^{k} (1-p_A)^{\#M_A - k}.
  \end{aligned}
\end{equation}

\begin{example}
  \label{eg:simulated_expected_dowker}
  Convergence of the actual weights to the expected weight computed by Equation \eqref{eq:dowker_histogram} is quite rapid if $p_0$ and $p_A$ are not close to $0.5$.  Consider a dataset with $8$ messages collected from $1000$ files, in which $5$ of the messages occur with probability $0.4$ and the remaining $3$ messages occur with probability $0.25$.  Such a dataset is much smaller than the dataset we ultimately considered!  To assess the variability in the resulting Dowker histogram, we simulated $300$ cases of this kind of dataset.  The resulting histograms are shown (in aggregate) in blue in Figure \ref{fig:expected_dowker}, over which the expected weights are plotted in red.  There is quite close agreement between the simulated and expected histograms.  
\end{example}

\begin{figure}
  \begin{center}
    \includegraphics[width=3.5in]{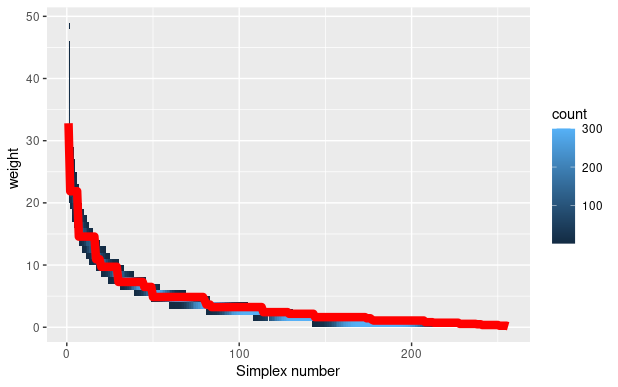}
    \caption{Histogram of expected Dowker weights (red) versus $300$ random trials (blue/gray) of $1000$ files, with a total of $8$ messages: $3$ messages with probability $0.25$, and $5$ messages with probability $0.4$.  Each point along the horizontal axis specifies a simplex (ordered so that their corresponding weights decrease).}
    \label{fig:expected_dowker}
  \end{center}
\end{figure}

\subsection{File ratios for two dialects}

Suppose that the dataset contains files from two dialects, $A$ and $B$.
Although files from both dialects might exhibit a given message pattern, this pattern may occur more frequently for files from one dialect.
Another way to interpret this situation is that files of one dialect may be more prevalent on certain Dowker simplices than on others.  If the distribution of files across the simplices differs, then it is possible to separate (some portion of) the dialects.  There can be simplices where the dialects overlap, namely certain patterns of messages that are exhibited with roughly equal probabilities by both dialects. These files cannot be separated using message patterns, and are of potential interest to file format analysts.

The ratio of dialect $B$ files to dialect $A$ files for a message pattern $K$ is
\begin{equation*}
  \frac{P(K|B) (\# B\text{ files})}{P(K|A) (\#A\text{ files})} = \left(\frac{P(K|B)}{P(K|A)}\right)\left(\frac{P(B)}{P(A)}\right)\left(\#\text{files}\right).
\end{equation*}
Notice that the ratio of conditional probabilities is the only dependence on the simplex $K$.  This ratio can be estimated using Equation \eqref{eq:message_prob}.

\begin{lemma}
  \label{lem:file_ratio}
  If $M_A \cap M_B = \emptyset$, then the ratio of conditional probabilities for $B$ files to $A$ files in simplex $K$ is expected to be
  \begin{equation*}
    \begin{aligned}
      \frac{P(K|B)}{P(K|A)} =& \left[\left(\frac{p_0}{p_A}\right)\left(\frac{1-p_A}{1-p_0}\right)\right]^{\#(K\cap M_A)} \times \\
      & \left[\left(\frac{p_B}{p_0}\right)\left(\frac{1-p_0}{1-p_B}\right)\right]^{\#(K\cap M_B)} \times \\
    &\left(\frac{1-p_0}{1-p_A}\right)^{\#M_A}
    \left(\frac{1-p_B}{1-p_0}\right)^{\#M_B}.
    \end{aligned}
  \end{equation*}
\end{lemma}
If the conditional independence assumption is violated, the ratio of conditional probabilities can still be computed.  In this case, it is usually called a \emph{pseudolikelihood ratio} \cite{arnold1991pseudolikelihood}.

\begin{proof}
  This is an elaborate calculation following from Equation \eqref{eq:message_prob}, the main goal of which is to eliminate the complements where they appear on each of $K$, $M_A$ and $M_B$.
  The following Boolean algebraic identities help to simplify the work:
  \begin{align*}
    K\cap M_A^c &= (K\cap M_A^c\cap M_B^c) \cup (K \cap M_A^c \cap M_B) \\
    K\cap M_B^c &= (K\cap M_A^c\cap M_B^c) \cup (K \cap M_A \cap M_B^c) \\
    K^c\cap M_A^c &= (K^c\cap M_A^c \cap M_B^c) \cup (K^c \cap M_A^c\cap M_B) \\
    K^c\cap M_B^c &= (K^c\cap M_A^c \cap M_B^c) \cup (K^c \cap M_A \cap M_B^c) 
  \end{align*}
  Specifically, the first two identities, followed by an application of the disjointness $M_A \cap M_B = \emptyset$, establishes that
  \begin{align*}
    \frac{p_0^{\#(K\cap M_B^c)}}{p_0^{\#(K\cap M_A^c)}} &= p_0^{\#(K \cap M_A \cap M_B^c) - \#(K \cap M_A^c \cap M_B)}\\
      &= p_0^{\#(K\cap M_A) - \#(K\cap M_B)}.
  \end{align*}
  In a similar way, we can derive that
  \begin{align*}
    \frac{(1-p_0)^{\#(K^c\cap M_B^c)}}{(1-p_0)^{\#(K^c\cap M_A^c)}} &= (1-p_0)^{\#(K^c \cap M_A \cap M_B^c) - \#(K^c \cap M_A^c \cap M_B)}\\
      &= (1-p_0)^{\#(K^c\cap M_A) - \#(K^c\cap M_B)}.
  \end{align*}
  Reorganizing yields the ratio
  \begin{equation*}
  \begin{aligned}
    \frac{P(K|B)}{P(K|A)} =& \left(\frac{p_0}{p_A}\right)^{\#(K\cap M_A)}
    \left(\frac{p_B}{p_0}\right)^{\#(K\cap M_B)} \times\\ &
    \left(\frac{(1-p_0)}{(1-p_A)}\right)^{\#(K^c\cap M_A)}
    \left(\frac{(1-p_B)}{(1-p_0)}\right)^{\#(K^c\cap M_B)}.
  \end{aligned}
  \end{equation*}
  To remove the remaining complements on the $K$, observe that
  \begin{align*}
    \#(K^c\cap M_A) &= \#M_A - \#(K\cap M_A)\\
    \#(K^c\cap M_B) &= \#M_B - \#(K\cap M_B)
  \end{align*}
  from which the desired result follows.
\end{proof}

Under the assumption that both $p_A$ and $p_B$ are greater than $p_0$, the first factor in the statement of Lemma \ref{lem:file_ratio} is less than $1$, while the second is greater than $1$.  Therefore, the ratio of dialect $B$ to dialect $A$ files in the simplex corresponding to $K$ is increased by ensuring that the messages in $K$ contain all of $M_B$ and none of $M_A$.  Messages outside $M_A \cup M_B$ do not impact the ratio of files in $K$ at all.  This is sensible: if one wishes to collect mostly dialect $B$ files, one looks for those that produce any messages in $M_B$ but none in $M_A$.

Conversely, places where the file ratio is close to $1$ are message patterns that are ambiguous.
Format analysts should spend more time on that particular set of files, since it is hard to disposition the files as clearly one dialect or the other without inspecting the file contents directly.
One could imagine that it might also be possible to craft new messages that discriminate between the dialects better by considering \emph{only} the files exhibiting these ambiguous message patterns.  This may greatly reduce the number of files that need to be considered.

\begin{corollary}
  \label{cor:file_ratio}
  If $M_A \cap M_B = \emptyset$ and $p=p_A=p_B$, then the ratio of $B$ files to $A$ files in simplex $K$ is expected to be
  \begin{equation*}
    \begin{aligned}
      \frac{P(K|B)}{P(K|A)} =
      & \left[\left(\frac{p}{p_0}\right)\left(\frac{1-p_0}{1-p}\right)\right]^{\#(K\cap M_B) - \#(K \cap M_A)} \times \\
    &\left(\frac{1-p}{1-p_0}\right)^{\#M_B-\#M_A}.
    \end{aligned}
  \end{equation*}
\end{corollary}

Corollary \ref{cor:file_ratio} is straightforward to apply and gives fine-grained information about where to look for files of a certain dialect.

\begin{example}
  \label{eg:lattice_example}
    Consider a notional dataset containing files from two dialects $A$ and $B$.  Suppose that there are three messages in total with $\#M_A = \#M_B=1$, and $M_A \cap M_B = \emptyset$.  If we let $p=p_A=p_B=0.4$, $p_0=0.2$, this dataset will satisfy the hypotheses of Corollary \ref{cor:file_ratio}.

  For three messages, there are $8$ possible message patterns.  These message patterns can be organized in a lattice based upon incrementally adding new messages, as is shown in Figure \ref{fig:lattice_example}.  In the figure, the message patterns are indicated by three squares: a filled square indicates that the corresponding message occurred, while an empty one indicates that the message is absent.  Message count increases as one moves up in the diagram: the bottom row corresponds to files exhibiting no messages, while the top row corresponds to files exhibiting all three messages.  

  Corollary \ref{cor:inconsistent_edges} asserts that because $p, p_0 < 0.5$, the per-dialect weights (shown as areas in the pie charts in Figure \ref{fig:lattice_example}) must decrease as one follows the gray arrows upwards in the diagram (adding messages incrementally).  Corollary \ref{cor:file_ratio} was used to compute the expected file ratio for each message pattern.
  
  The largest weight is on the bottom of the diagram, where the files exhibit no messages.  Since the ratio is $1$ for that message pattern (the pie chart shows equal areas for both dialects), it is not possible to separate dialects for the files that exhibit no messages.  One might imagine that there are other features that might be informative, but that these are simply not captured by the three messages under consideration.

  The next largest weight is in the second row, where just the messages in $M_A$ or $M_B$ occur.  (If $p>0.5$ instead, the largest weights can occur on rows other than the bottom, because the hypotheses of Corollary \ref{cor:inconsistent_edges} are not satisfied.)
  
  Considering the pie charts, the majority of the dialect $A$ files are on the left of the diagram, while the majority of the dialect $B$ files are on the right of the diagram.  Therefore, coarse dialect separation can be done on those files on the left or right, on account of their message patterns being different.
\end{example}

\begin{figure}
  \begin{center}
    \includegraphics[width=2.75in]{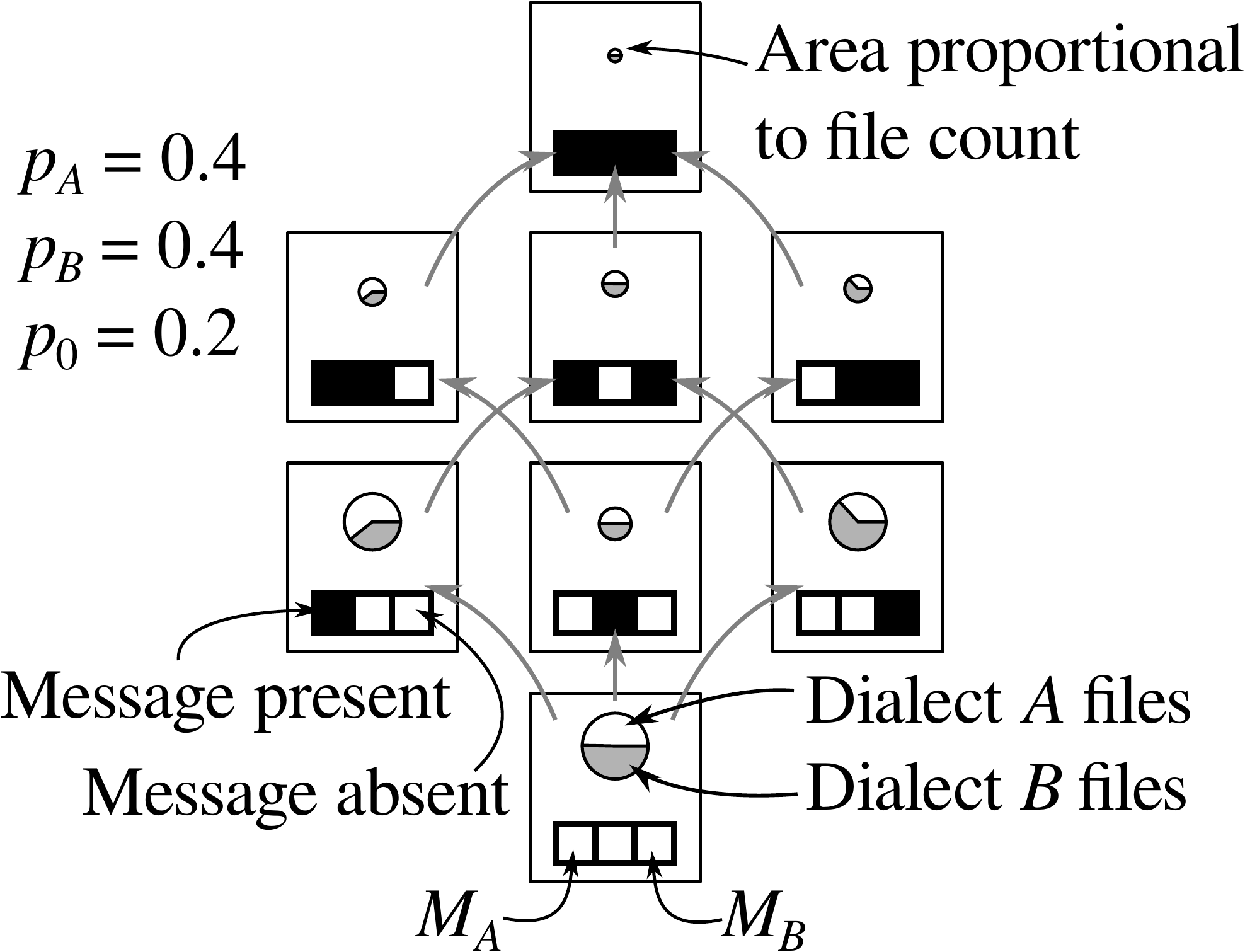}
    \caption{Expected weight and file ratio for two dialects using three messages as described in Example \ref{eg:lattice_example}.}
    \label{fig:lattice_example}
  \end{center}
\end{figure}

\subsection{Separating dialects by thresholding posterior probabilities}

Of course, the disjointness of $M_A$ and $M_B$ required by Lemma \ref{lem:file_ratio} might not hold exactly in practice.
Messages that lie in the intersection of $M_A$ and $M_B$ will not impact the ratio of files on any simplex if $p_A \approx p_B$, as is also easily seen in Corollary \ref{cor:file_ratio}.
However, if $p_A$ and $p_B$ differ substantially, this will tend to change the ratio of files from the estimate in Lemma \ref{lem:file_ratio} for any simplex that contains any messages in $M_A \cap M_B$.
Furthermore, if there are not many $B$ files, it may be difficult to estimate $p_B$ or the true contents of $M_B$.

Observe that files in the tail of the Dowker histogram are also in the tail of the message count histogram.  These are instances where $P(K|A)$ or $P(\#K=n|A)$ is low.  On the other hand, if $P(K|B)$ or $P(\#K=n|A)$ is comparatively higher, this will cause the ratio of dialect $B$ files to be statistically significantly higher, and thereby possible to detect.

A systematic way to exploit this information is to consider the probability that a given file is from a certain dialect, given that it exhibits a certain message pattern.  This is called the \emph{posterior probability}, and can be computed using Bayes' theorem.  For instance, to ascertain the probability that a file is in dialect $A$ given that it produced message pattern $K$, this is given by
\begin{equation}
  \label{eq:posterior_prob}
  P(A|K) = P(K|A) \frac{P(A)}{P(K)}.
\end{equation}
This formula defines a test statistic---a quantity that yields a classifier for dialect $A$ files upon thresholding $P(A|K)$.  Files with large $P(A|K)$ are likely from dialect $A$, while files with a lower value of $P(A|K)$ are less likely to be from dialect $A$.

In order to compute $P(A|K)$, one needs to obtain each factor in Equation \eqref{eq:posterior_prob}.  $P(K)$ is easily assessed by computing the frequency of the message pattern $K$ in the dataset at hand.  $P(K|A)$ is given by Equation \eqref{eq:message_prob}, subject to the assumptions mentioned earlier in the paper.  Moreover, if one has a training set consisting (almost) entirely of dialect $A$ files, then one can estimate $P(K|A)$ by simple counting.  This avoids the potential issues with violations of the independence of messages.

Finally, $P(A)$ is the expected probability that a file will be of dialect $A$, given no further information.  This last factor is the most difficult to estimate, and can be best thought of as a ``risk factor'': choosing a larger value of $P(A)$ will result in a higher estimate for $P(A|K)$, while choosing a lower value of $P(A)$ will consequently reduce the estimate for $P(A|K)$.  Thus, if dialect $A$ files are dangerous, it is wise to overestimate $P(A)$.  Conversely, if $A$ files are likely benign, an underestimate of $P(A)$ will produce fewer false alarms.

%%%%%%%%%%%%%%%%%%%%%%%%%%%%%%%%%%%%%%%%%%%%%%%%

\section{Results}

Figure \ref{fig:goodA_mp} shows the message probabilities for the SafeDocs Evaluation 3 Universe A good files dataset.  Some messages occur more than 50\% of the time; for these messages, it is more useful to assume that their \emph{absence} is an informative event.  Therefore, as a preprocessing step, when computing the probability of a message $k$, if its probability is $p_k > 0.5$, we instead use $1-p_k$ in what follows.

The vast majority of files have a low probability of occurrence, while there are roughly $6$ messages with an elevated probability (note the logarithmic scale).  Using the model we propose is tantamount to discretizing the probabilities shown to two separate levels: a higher one for the first $6$ messages, and a lower one for the rest of the messages.  The average probability for each of the first $6$ messages is $p_{\text{good}} = 0.380$.  The specific messages this threshold selects are shown in Table \ref{tab:goodA_mA}, so a reasonable cutoff for the probability of an $M_{\text{good}}$ message is $p_{\text{good}}>0.25$.  (Note that Table \ref{tab:goodA_mA} shows the raw probabilities of the messages, which may exceed $0.5$.)  The regexes for these messages appear in Table \ref{tab:weighty_messages}.  It happens that nearly all of the $M_{\text{good}}$ messages are nonzero exit codes for parsers, but with no further detail.  Four of the messages are from {\tt caradoc}, which is known to be a fairly stringent parser.  One may interpret $M_{\text{good}}$ as consisting of mostly benign messages that are indicative of otherwise good files.

Running the same process on the SafeDocs Evaluation 3 Universe A bad files dataset with same threshold as before, $p_{\text{bad}}>0.25$, yields $54$ messages with an elevated probability, at approximately $0.312$.  The intersection between this set and the corresponding set of messages for the good files is nonempty and consists of $4$ messages, shown in the upper portion of Table \ref{tab:goodA_mA}.  The actual regexes appear in Table \ref{tab:weighty_messages}.

When there are many messages, the estimates of $P(K|A)$ tend to be very small, regardless of the dialect, and therefore are subject to substantial sampling error.  In the Universe A good files set, none of the files produced no messages (after inverting the meaning of any message with probability greater than $0.5$).  Therefore Equation \eqref{eq:p0_estimator} cannot be used.  Moreover, many of the messages with low probabilities are not independent.  For these reasons, it is better to select an overestimate for $p_0$, since this results in more frequent co-occurrence between messages.  The largest overestimate for $p_0$ is the value at the threshold chosen for $M_{\text{good}}$, namely $p_0=0.25$, as shown on Figure \ref{fig:goodA_mp}.  This choice will be later confirmed by the agreement between the actual and expected Dowker histograms described in the next Section by considering Figure \ref{fig:goodA_histogram}.

\begin{figure*}
  \begin{center}
    \includegraphics[width=5in]{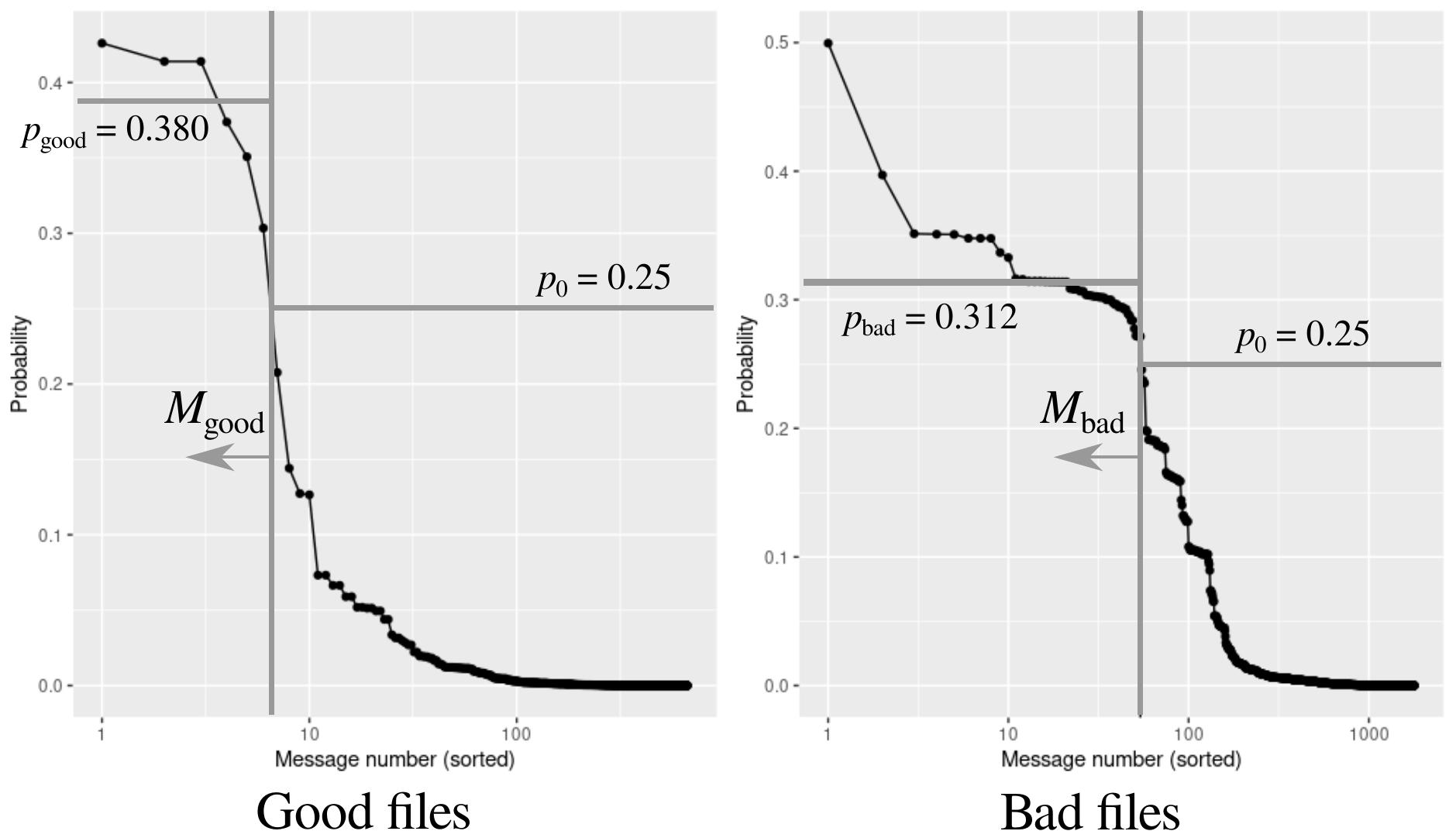}
    \caption{Probability a given message will occur for the SafeDocs Evaluation 3 Universe A good files (left) and bad files (right).  In both frames, messages in $M_A$ are to the left of the gray vertical line.  Estimates for $p_A$ and $p_0$ are shown as horizontal lines.}
    \label{fig:goodA_mp}
  \end{center}
\end{figure*}

\begin{table}
  \begin{center}
    \caption{Messages in $M_{\text{good}}$ for Universe A files}
    \label{tab:goodA_mA}
    \begin{tabular}{|c|l|c|c|}
      \hline
      Message&Parser and options&Prob. in&Prob. in\\
      &&good files&bad files\\
      \hline
      \hline
      1&{\tt caradoc extract}&0.414&0.697\\
      \hline
      122&{\tt caradoc stats}&0.414&0.697\\
      \hline
      943&{\tt origami pdfcop}&0.426&0.500\\
      \hline
      2055&{\tt qpdf}&0.303&0.603\\
      \hline
      \hline
      243&{\tt caradoc stats --strict}&0.626&0.842\\
      \hline
      334&{\tt caradoc stats --strict}&0.351&0.033\\
      \hline
    \end{tabular}
  \end{center}
\end{table}

\subsection{Message pattern distribution}

\begin{figure*}
  \begin{center}
    \includegraphics[width=5in]{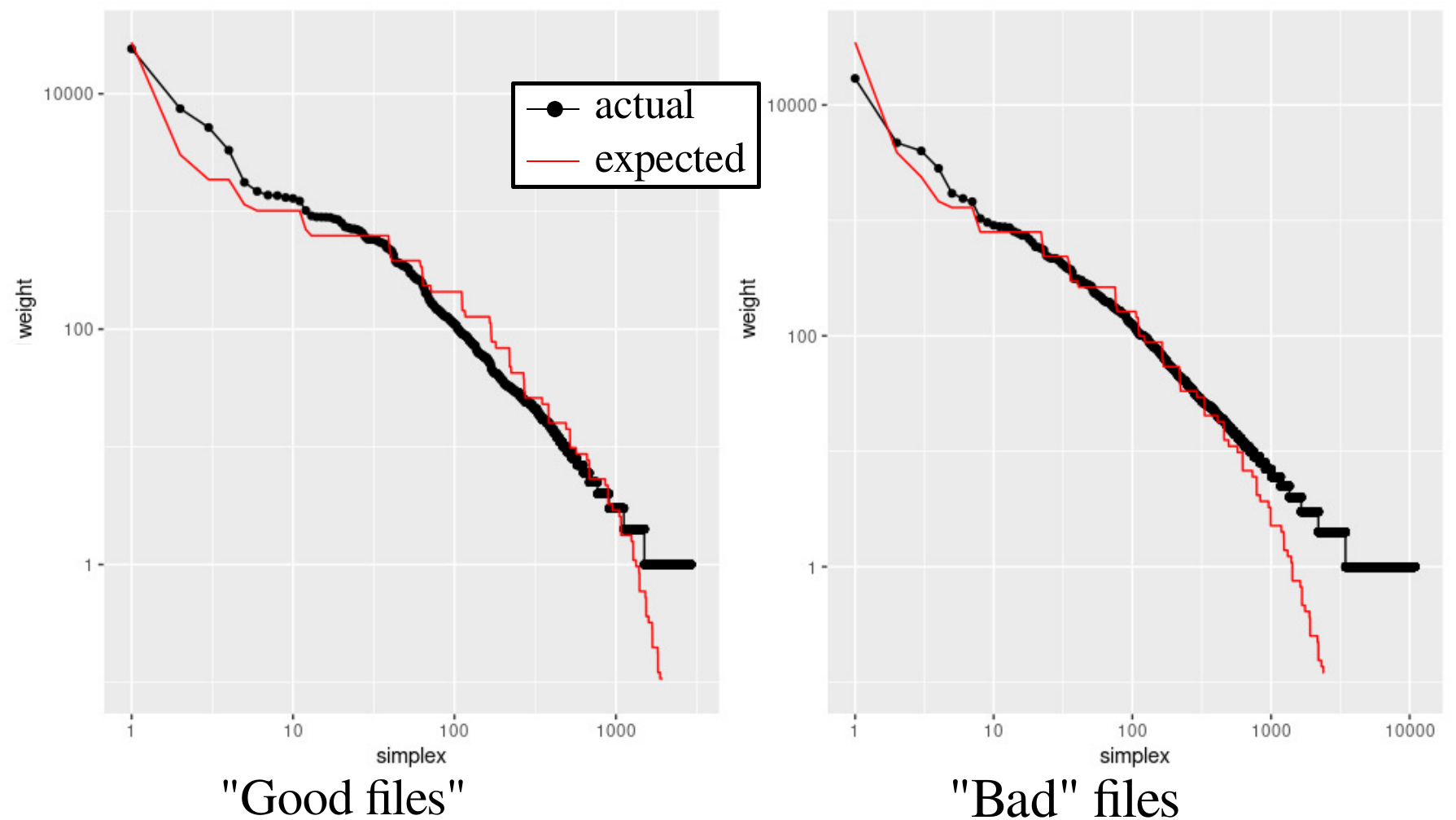}
    \caption{Histograms of expected weights (file counts) versus actual weights for the SafeDocs Evaluation 3 Universe A datasets: (left) the good file subset, used for training, and (right) the bad files subset.}
    \label{fig:goodA_histogram}
  \end{center}
\end{figure*}

  There are $2^{\#M}$ possible message patterns.  Since the total number of messages collected was large, it is not feasible to compute the expected weights for each message pattern.  Instead, it is much more practical to compare only the message patterns that were actually observed.  This means that we need to normalize the probabilities so that the sum over all \emph{observed} message patterns is $1$, instead of normalizing so that the sum over \emph{all possible} message patterns is $1$.  This being done, the comparison between observed and expected weight distributions is shown at left in Figure \ref{fig:goodA_histogram}.  There is close agreement over the entire distribution, which can be taken as validating our model of the data, and validating our choice of $p_0$ in particular.

  We can also compare the expected distribution of weights with the bad files set.  It is most interesting to do this not with $M_{\text{bad}}$ and $p_{\text{bad}}$, but rather with the parameters $M_{\text{good}}$, $p_{\text{good}}$, and $p_0$.  Because of the close agreement on the good files, differences between the expected and observed distributions in the bad files are more visible.  The resulting plot is shown at right in Figure \ref{fig:goodA_histogram}.  While there is still close agreement for the large weights (indeed, there is somewhat better agreement than for the good files), the message patterns with low weights are quite different.  The expectation is that there would be rather fewer files exhibiting particular message patterns than actually occurred.  Since low weights are correlated with higher message counts, this suggests that the bad files frequently produce messages that are not in $M_{\text{good}}$.  In other words, the bad files consist of a distinct dialect from the good files.

  \subsection{Behavior differences within the datasets}

  Grouping files based upon message patterns does yield a useful tool for clustering related behaviors.  This is exhibited in both the good and bad files, with the largest weights corresponding to clear behavioral patterns.

  The message patterns with the largest weights are shown in Table \ref{tab:good_weighty_simplices}.  Since messages $243$, $991$, $1319$, and $287$ had probabilities exceeding $0.5$, we will consider their absence rather than their presence as noted earlier.  For instance, {\tt pdfium} always produces message $991$, so it is rather uninformative.  The messages that were initially reported (without this reversal) are shown in the column ``Raw messages'', while the messages after this reversal are recorded in the column ``Messages (corrected)''.

  Although these messages are recoverable (since the files under consideration are in the good files set),
  these message patterns correspond to three distinct classes of behavior: errors within a compressed stream, syntax errors within the PDF file, and type errors.
  
  The specific messages involved in Table \ref{tab:good_weighty_simplices} are shown in Table \ref{tab:weighty_messages}.
  Several of the regexes are duplicated because the same base parser {\tt caradoc} can be run with different options.
  (These duplicated messages are not independent.)

  \begin{table*}
    \begin{center}
      \caption{The message patterns with the largest weights (file counts) in good files}
      \label{tab:good_weighty_simplices}
      \begin{tabular}{|c|l|l|c|l|}
        \hline
        Weight &Raw messages & Messages (corrected)&Message count& Error taxonomy\\
        \hline
        \hline
          24101&
          334, 943, 991, 1319, 2287&
          243, 943, 334&
          3&
Compressed stream error \\
\hline
7470&
334, 991, 1319, 2287&
243, 334&
2&
Compressed stream error \\
\hline
5170&
243, 258, 991, 1319, 2287&
258&
1&
Syntax error (lexing)\\
\hline
3313 &
351, 991, 1319, 2287&
243, 351&
2&
Syntax error (newline placement)\\
\hline
1767&
1, 19, 122, 140, 243, 330, 991, &
1, 19, 122, 140, 330&
5&
Type error \\
 &1319, 2287&&& \\
\hline
\end{tabular}
\end{center}
  \end{table*}

  For the bad files, the message patterns with the largest weights involve substantially more messages, as shown in Table \ref{tab:bad_weighty_simplices}.  Again, several messages occur with probability greater than $0.5$ and so their absence is shown in the column ``Messages (corrected)''.  These message are different from the good files, and are $1$, $122$, $243$, $991$, $1319$, $2055$, and $2287$.  Without this correction, the most common message pattern in the bad files is the same as one occurring in the good files, involving a syntax (lexing) error.  However, without the correction, it is rather different.  This message pattern contains $5170/(5170+16980)= 23\%$ good files and $77\%$ bad files.  Further delineation of this particular syntax error as recoverable or not is impossible given the messages we collected.

  Aside from the most common pattern, there is no overlap between good and bad files among the most common message patterns.  In the remaining message patterns most commonly exhibited by the bad files, there are several patterns corresponding to damaged {\tt xref} tables.  Since {\tt xref} tables provide important structural information about the contents a PDF file, it is unsurprising that damage to the {\tt xref} table results in many more messages being produced.

    \begin{table*}
    \begin{center}
      \caption{The message patterns with the largest weights (file counts) in bad files}
      \label{tab:bad_weighty_simplices}
      \begin{tabular}{|c|l|l|c|l|}
        \hline
        Weight & Raw messages & Messages (corrected)& Message count&Error taxonomy\\
        \hline
        \hline
          16980&
          243, 258, 991, 1319, 2287&
          1, 122, 258, 2055&
          4&
          Syntax error (lexing)\\
          \hline

          4702&
          (not listed for space considerations)&
          (not listed for space considerations)&
          104&
          Damaged or missing {\tt xref} table
          \\
          \hline
         
          4000&
          (not listed for space considerations)&
          (not listed for space considerations)&
          84&
          Damaged or missing {\tt xref} table
          \\
          \hline
     
          2825&
          (not listed for space considerations)&
          (not listed for space considerations)&
          72&
          Damaged or missing {\tt xref} table
          \\
          \hline

          1715&
          243,271,991,1319,2055,2287&
          1,122,271&
          3&
          Syntax error
          \\
          \hline
\end{tabular}
\end{center}
  \end{table*}

  \begin{table*}
    \begin{center}
      \caption{Messages involved in Tables \ref{tab:goodA_mA}, \ref{tab:good_weighty_simplices}, and \ref{tab:bad_weighty_simplices}}
      \label{tab:weighty_messages}
      \begin{tabular}{|c|l|l|}
        \hline
        Message & Parser & {\tt stderr} regex\\
        \hline
        \hline
        1&
        {\tt caradoc extract}&
        (exit code indicating error) \\ % (none)
        \hline
        19&
        {\tt caradoc extract} &
        {\tt Type error : Unexpected entry .* in instance of class .* in object .* !}\\ % Addition, PDF indirect objects, Type error
        \hline
        122&
        {\tt caradoc stats} & 
        (exit code indicating error) \\ % (none)
        \hline
        140&
        {\tt caradoc stats} & 
        {\tt Type error : Unexpected entry .* in instance of class .* in object .* !} \\ % Addition, PDF indirect objects, Type error
        \hline
        243&
        {\tt caradoc stats --strict} &
        (exit code indicating error) \\ % (none)
        \hline
        258&
        {\tt caradoc stats --strict}&
        \verb;PDF error : Lexing error : unexpected character : 0x[A-Fa-f\d]+ at offset;\\ 
        &&\verb;\d+ \[0x[A-Fa-f\d]+\] in file !;\\ % Addition, PDF file structure, Syntax error
        \hline
        271&
        {\tt caradoc stats --strict}&
        \verb;PDF error : Syntax error at offset \d+ \[0x[A-Fa-f\d]+\] in file !;\\ % Misformation, Syntax error, PDF file structure
        \hline
        330&
        {\tt caradoc stats --strict} & 
        {\tt Type error : Unexpected entry .* in instance of class .* in object .* !} \\ % Addition, PDF indirect objects, Type error
        \hline
        334&
        {\tt caradoc stats --strict} &
        {\tt Warning : Flate\/Zlib stream with appended newline in object .*} \\ % Value error, Addition, PDF streams
        \hline
        351&
        {\tt hammer} &
        \verb;VIOLATION\[\d+\]@\d+ \(0x[A-Fa-f\d]+\): No newline before 'endstream'; \\ 
        && \verb;\(severity\=.*\);\\ % Omission, PDF streams, Syntax error
        \hline
        943&
        {\tt origami pdfcop} &
        (exit code indicating error)\\ % (none)
        \hline
        991&
        {\tt pdfium} &
        \verb;Processed \d+ pages\.; \\ % General informative
        \hline
        1319&
        {\tt peepdf} & 
        (exit code indicating error) \\ % (none)
        \hline
        2055&
        {\tt qpdf} &
        (exit code indicating error)\\ % (none)
        \hline
        2287&
        {\tt verapdf pdfbox} & 
        (exit code indicating error) \\ % (none)
        \hline
      \end{tabular}
    \end{center}
  \end{table*}

  From a careful inspection of the first two rows of Table \ref{tab:good_weighty_simplices}, we can conclude that there is a violation of Corollary \ref{cor:inconsistent_edges}.  That is, the addition of message $943$ resulted in a higher weight with than without it.  This indicates that there is a strong relationship between message $943$ and the messages which indicate issues with compressed streams. We can infer that the presence of message $943$ is sometimes indicative of problems with compressed streams in PDFs, even though the message is merely an exit code.

  A violation of Corollary \ref{cor:inconsistent_edges} was also exhibited by the bad files as well, though it does not appear in Table \ref{tab:bad_weighty_simplices} because the lists of messages are too long to fit.  This violation effectively groups together a collection of messages related to broken {\tt xref} tables.
    
  \subsection{Interactive display of Dowker complexes}

  \label{sec:dowker_display}

  While the Dowker complex can be computed succinctly using R as discussed in Section \ref{sec:mp_and_dowker}, this implementation is not particularly efficient for large datasets.  Additionally, it does not easily support visualization of the lattice structure mentioned in previous sections.  Therefore, we developed an optimized Python version of the Dowker complex construction.  For a small number of messages, a 2d representation of the lattice structure (like what appears in Figure \ref{fig:lattice_example}) suffices, but this becomes increasingly cluttered with more messages.

  To remedy this issue, we implemented a Python version that embeds the Dowker complex in $3$ dimensions and permits interactive examination. The Dowker visualization relies on {\tt Plotly} and its 3d network graph \cite{plotly}. The 3d network graph consists of connected nodes.  Each node corresponds to a message pattern whose weight exceeds a user-chosen parameter, and each edge corresponds to the addition of a single message.  The resulting graph can be customized by setting node and edge colors, positions etc.

  % The graph requires that we enter node attributes in separate lists, one entry per node, for $x$-coordinate, $y$-coordinate, $z$-coordinate, size, color, etc, and then edge attributes as a lists of pairs of (edge origin $x$-coordinate, edge destination $x$-coordinate), etc, and list of edge colors.

We start with the Boolean matrix representation of the data stored in an array  {\tt msgMatrix} in which each file is a row and each message is a column (like in the R implementation, this is the transpose of the matrices shown in Figure \ref{fig:uniA_matrices}).  Each row corresponds to a message pattern, which will be displayed as a node in the graph upon the removal of duplicates.  The first step is to construct the mappings of nodes to attributes as well as getting all possible connected nodes to be used to identify edges.

For efficiency, our implementation uses Python's {\tt hash()} function to quickly and uniquely identify each message pattern.  Importantly, this allows us to define a function {\tt getConnNodes()} that takes a message pattern and computes the hashes for all possible message patterns with one fewer message.  Using this function {\tt getConnNodes()}, the pseudocode below shows how to construct the Dowker complex and its weights.

\begin{lstlisting}[language=python]
for row in msgMatrix:
 rowHash = hash(str(row))
 label = str(row)
 if rowHash in nodeWeightMap:
  nodeWeightMap[rowHash]+=1
 else: 
  nodeWeightMap[rowHash]=1
  nodeLabelMap[rowHash] = label
  
  # Find labels for all possible
  # connected nodes, by finding
  # all nodes with 1 less message
  nodeConnNodeMap[label] =
    getConnNodes(row)
\end{lstlisting}

The above code is only notional because we found that the call to {\tt hash(str(row))} turned out to be about $825$ times slower than using {\tt numpy} builtins to first interpret the row as bytes and then hexify into a string. This is because using {\tt str()} on an array row calls {\tt numpy}'s {\tt array2str} with internal recursion that incurs about $100$ operations per row and $20$ operations per element while {\tt numpy.packbits()} only incurs about $6$ operations per row and $0$ operations per element. Combined with other optimization efforts, the Dowker generation was sped up by a factor of $628$ over the na\"ive translation of the pseudocode.

Our preferred layout is a layered one, where each layer consists of all nodes with the same number of messages, and each layer is arranged in a circle.  The layers are sorted numerically by message count. Other visualization methods we have tried include laying out nodes in a force-directed Kamada-Kawai and Fruchterman-Reingold \cite{igraph}, though the renderings these produced were generally harder to interpret because they disrupted the layered structure.

\begin{figure}
  \begin{center}
    \includegraphics[width=3.5in]{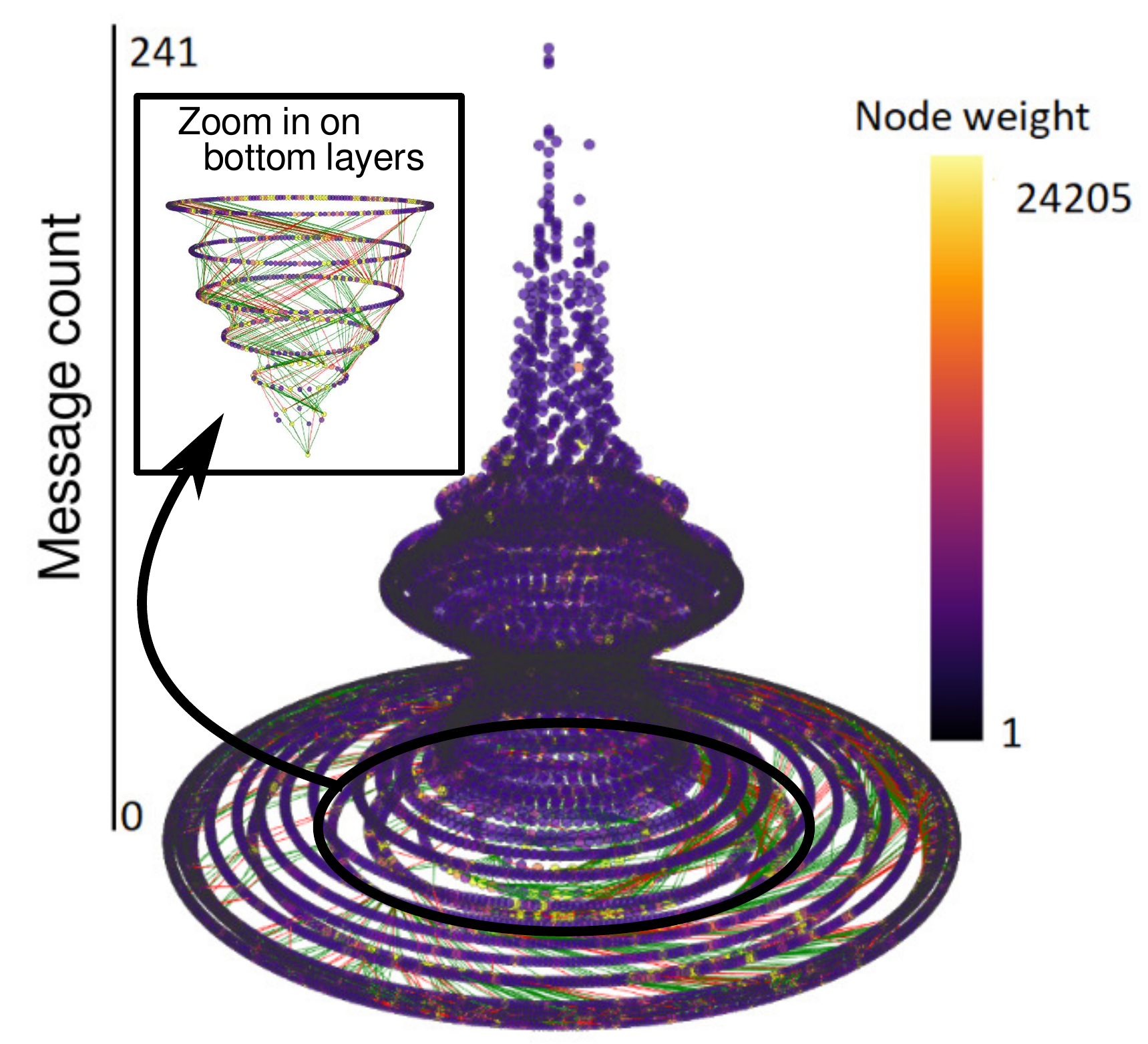}
    \caption{3d rendering of the Dowker complex for the SafeDocs Evaluation 3 Universe A, colored by weight.  Each point corresponds to a message pattern, with the number of messages increasing as one moves up in the diagram.  Edges marked in red correspond to violations of Corollary \ref{cor:inconsistent_edges}. Inset shows the ten ``layers'' corresponding to fewer than $10$ messages.}
    \label{fig:3ddowker}
  \end{center}
\end{figure}

Figure \ref{fig:3ddowker} shows the Dowker graph generated from the SafeDocs Universe A combined dataset, with nodes colored such that higher weight nodes (more files triggering those message patterns) are colored yellow, and lower weight nodes are colored purple.  The majority of files have few messages, as indicated by the wide ``base'' at the bottom of the rendering, where the brighter colored nodes are located.  Nodes become more sparse at higher layers, corresponding to files that produced more messages.  The wider ``neck'' in the the middle indicates that many message patterns had a message count around $100$.  Figure \ref{fig:3ddowker} also shows edges connecting neighboring message patterns which differ by a single message.  The edges are colored in either green to indicate that the weight decreased in accordance with Corollary \ref{cor:inconsistent_edges}, or red if the weight increased.  Because Corollary \ref{cor:inconsistent_edges} depends on the conditional independence of messages, red edges indicate that this assumption has failed for the message patterns involved.  Highly dependent messages are often indicative of dialect boundaries, so they could be candidates for further analysis based upon file contents.  Moreover, the sparsity of edges in the upper portion of the diagram indicates that most of the message patterns for the associated files files are unrelated to one another; adding or removing a single message drastically reduces the number of files exhibiting that new pattern.

\begin{figure}
  \begin{center}
    \includegraphics[width=3.5in]{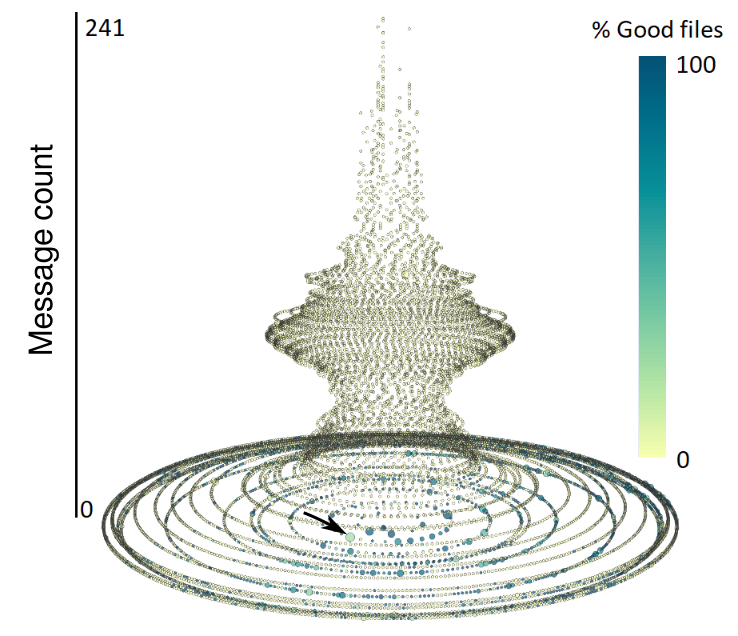}
    \caption{3d rendering of the Dowker complex for the SafeDocs Evaluation 3 Universe A, colored by percentage of good files. Each point corresponds to a message pattern.  The size of the points indicates the weight of each message pattern.  To reduce clutter, the edges are not shown.}
    \label{fig:universedowker}
  \end{center}
\end{figure}

We can also change the coloring of nodes to display the different dialects.
Figure \ref{fig:universedowker} shows the nodes' classification based on whether the files in question were good (blue), bad (yellow), or a mixture (shades of green).  It is clear that the message count is clearly indicative of whether a file is good or bad: the nodes with high message counts are all yellow, corresponding to an overwhelming majority of bad files.  Additionally, since the weights in Figure \ref{fig:universedowker} are shown by the size of the nodes, the message patterns with the largest weights shown in Table \ref{tab:good_weighty_simplices} for the good files are quite visible, and they all appear near the bottom of the diagram.  The largest weight message pattern for the bad files in Table \ref{tab:bad_weighty_simplices} is also visible near the bottom of the diagram as well, and is marked with an arrow.  The presence of this particular message pattern, and several other majority-bad nodes with low message count near the bottom of the diagram justifies the use of the Dowker complex for dialect classification, rather than using only message count.

  \subsection{Separating dialects by thresholding}
\label{sec:posterior_performance}
  Thresholding posterior probabilities works well for separating the good files from the bad files.
  Starting with the Universe A good files as a training set ensures that we have a training set that is mostly good files.

  From this training set, we estimate $P(K|\text{good})$ for all message patterns $K$ that are exhibited in the data.
  To this end, we can either use Equation \eqref{eq:message_prob} (theoretical) based upon our previous estimates of $M_{\text{good}}$ and $p_{\text{good}} = 0.380$, or we can compute $P(K|\text{good})$ directly by counting how many message patterns are exhibited (empirical).

  Now let us consider the combined dataset with two dialects, namely both the Universe A good and bad files, but let us ``forget'' which file comes from which set.
  Given the fact that we know how many files of each dialect there are (but not which file is which), we know that $P(\text{good}) = 0.5$, since the data happen to contain equal numbers of both files.
  For this combined dataset, we can estimate $P(K)$ directly from the data by counting the number of times each message pattern occurs (just as we did for $P(K|\text{good})$).

  Given all of these facts, we can then use Equation \eqref{eq:posterior_prob} to determine the probability $P(\text{good}|K)$ that a given file is good, given the particular message pattern $K$ produced by the file.  It still remains to select a probability threshold to use to determine whether we declare a file as good or bad.  For that threshold, we can determine the \emph{recall} (the fraction of good files with probability above our chosen threshold), and the \emph{precision} (the fraction of truly good files above our threshold versus the total number of files above our threshold).  An ideal classifier will have both precision and recall as close to $1$ as possible.  

  \begin{figure}
  \begin{center}
    \includegraphics[width=3in]{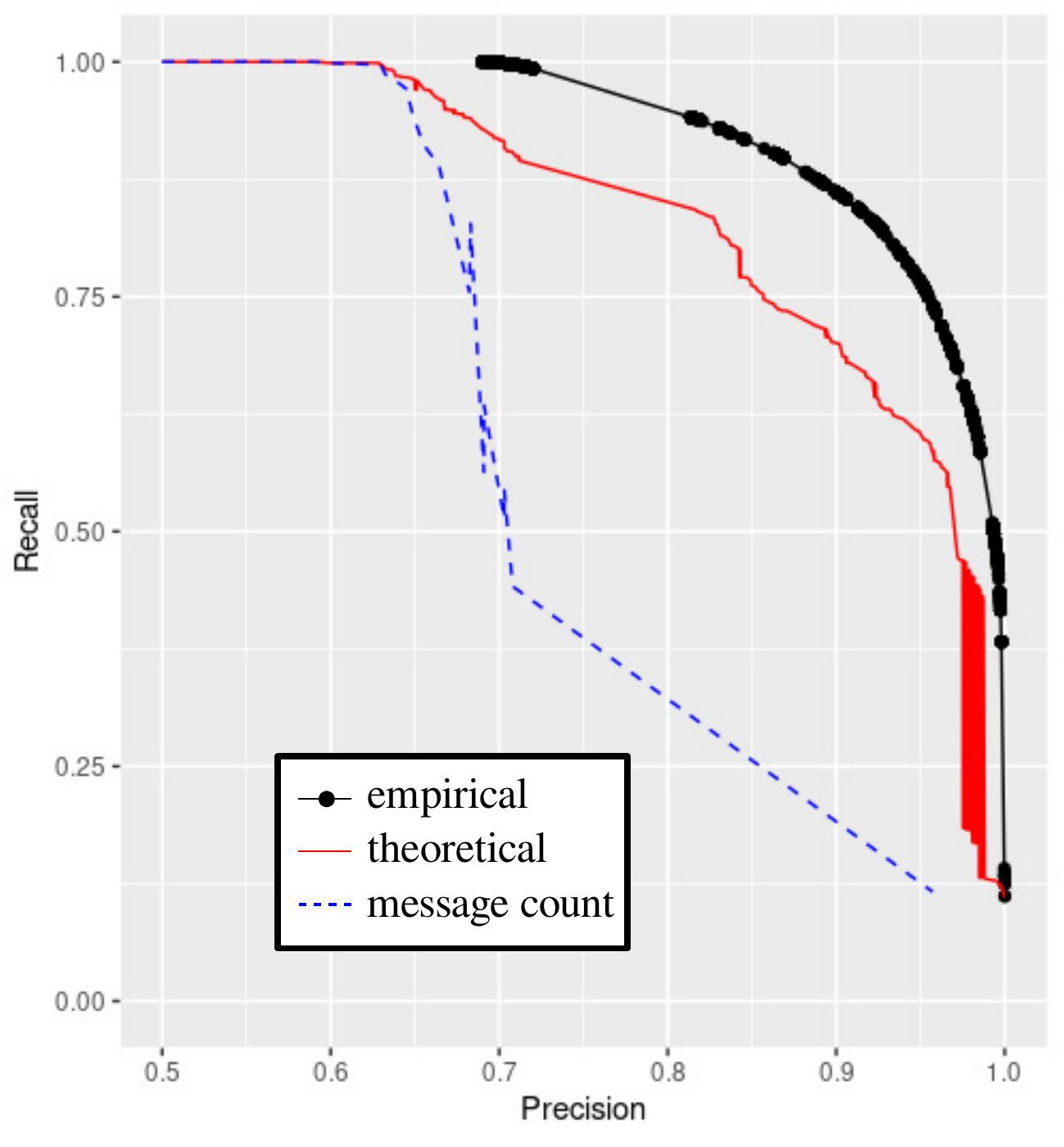}
    \caption{Precision versus recall for separating the SafeDocs Evaluation 3 Universe A good files from the bad files, estimating $P(K|\text{good})$ empirically from training on the good files (black), using Equation \eqref{eq:message_prob} with parameters as described in the text (red), or using message count alone (blue, dashed).}
    \label{fig:posterior_performance}
  \end{center}
\end{figure}

  Since we do not have any prior knowledge, the best measure of performance is to consider all possible thresholds, and to aggregate all precision and recall scores.  This is shown in Figure \ref{fig:posterior_performance}.  The figure shows three curves: the red curve uses Equation \eqref{eq:message_prob} to estimate $P(K|\text{good})$, the black curve uses the values estimated empirically from the good files alone, and the blue dashed curve shows the performance of classifying using message count alone.  While message count alone does not do a good job of classifying (largely because it misclassifies the bad files that produce only a few messages), the other two curves show good performance.   Using the empirical estimates yields the best overall performance.  This is not unexpected, since the conditional independence assumption made by Equation \eqref{eq:message_prob} does not entirely hold.

  From Figure \ref{fig:posterior_performance}, we conclude that the precision of our method typically exceeds its recall for most threshold choices.  One can interpret this to mean that the files with a high probability of being good based upon their message patterns are indeed good, though many good files are missed because they exhibit more unusual message patterns.  Intuitively, this means that many messages produced are of a benign nature.

%%%%%%%%%%%%%%%%%%%%%%%%%%%%%%%%%%%%%%%%%%%%%%%%

\section{Conclusion}

This paper provides a theoretical basis and practical algorithms for determining the format dialect of files within a dataset based upon the statistics of messages that parsers produce as they consume the files.  The methods we used are based upon thresholding the posterior probability of a file being in a certain dialect, using the idea that messages occur independently once they are conditioned upon dialect.
Even though a na\"ive classification of files based on message count might seem clearly the best, our method outperforms this by a wide margin.

Moreover, using our method, a format analyst can therefore greatly reduce the number of files they need to consider, by focusing their attention on \emph{only} the files exhibiting message patterns with an ambiguous posterior probability. 
By looking at only these files, one can likely discover features that serve as ``cut points'' between dialects.
Moreover, the theoretical file ratios allow one to predict which message patterns will be easy to disposition and which will not.
Those that are not easy to disposition will tell the format analyst about what kind of new messages need to be crafted to discriminate between dialects.
Such new messages are likely easier to construct under the condition that \emph{the ambiguous message pattern is already present}.  This may greatly reduce the number of files that need to be considered when constructing new messages.

Besides dispositioning of files as one dialect or another, the relationships between the message patterns themselves allow for a finer analysis.
Our theoretical model establishes that the number of files exhibiting a given pattern of messages should decrease as more messages are triggered (Corollary \ref{cor:inconsistent_edges}).
Violations of this result indicates places where our assumption of conditional independence is violated.
In our data, these violations allow one to draw inferences about the semantic meaning of certain parser exit codes that are not associated with human-readable regular expressions.

Conversely, it is sometimes a valuable exercise to craft intentionally ambiguous files, such as polyglot or schizophrenic files.
Polyglot files can be used to probe a format specification, as they often trigger unexpected corner cases in its logic.
Intuitively, polyglot files are easiest to construct when they elicit a pattern of messages with file ratio close to $1$.  Knowing which message patterns already have file ratios close to $1$ may aid in constructing these files.

\section*{Acknowledgments}

The authors would like to thank the SafeDocs test and evaluation team, including NASA (National Aeronautics and Space Administration) Jet Propulsion Laboratory, California Institute of Technology and the PDF Association, Inc., for providing the test data.  The authors would like to thank Denley Lam for the initial processing of the files into sets of messages.

This material is based upon work supported by the Defense Advanced Research Projects Agency (DARPA) SafeDocs program under contract HR001119C0072.  Any opinions, findings and conclusions or recommendations expressed in this material are those of the authors and do not necessarily reflect the views of DARPA.



\bibliographystyle{IEEEtran}
\bibliography{dowkerstat_bib}

\begin{thebibliography}{10}
\providecommand{\url}[1]{#1}
\csname url@samestyle\endcsname
\providecommand{\newblock}{\relax}
\providecommand{\bibinfo}[2]{#2}
\providecommand{\BIBentrySTDinterwordspacing}{\spaceskip=0pt\relax}
\providecommand{\BIBentryALTinterwordstretchfactor}{4}
\providecommand{\BIBentryALTinterwordspacing}{\spaceskip=\fontdimen2\font plus
\BIBentryALTinterwordstretchfactor\fontdimen3\font minus
  \fontdimen4\font\relax}
\providecommand{\BIBforeignlanguage}[2]{{%
\expandafter\ifx\csname l@#1\endcsname\relax
\typeout{** WARNING: IEEEtran.bst: No hyphenation pattern has been}%
\typeout{** loaded for the language `#1'. Using the pattern for}%
\typeout{** the default language instead.}%
\else
\language=\csname l@#1\endcsname
\fi
#2}}
\providecommand{\BIBdecl}{\relax}
\BIBdecl

\bibitem{Robinson_looking_2021}
M.~Robinson, ``Looking for non-compliant documents using error messages from
  multiple parsers,'' in \emph{LangSec 2021, a subconference of IEEE Security
  \& Privacy}, May 2021.

\bibitem{Ambrose_2020}
K.~Ambrose, S.~Huntsman, M.~Robinson, and M.~Yutin, ``Topological differential
  testing, {\tt arxiv:2003.00976},'' 2020.

\bibitem{belaoued2015real}
M.~Belaoued and S.~Mazouzi, ``A real-time {PE}-malware detection system based
  on chi-square test and {PE}-file features,'' in \emph{IFIP International
  Conference on Computer Science and its Applications}.\hskip 1em plus 0.5em
  minus 0.4em\relax Springer, 2015, pp. 416--425.

\bibitem{al2018ransomware}
B.~A.~S. Al-rimy, M.~A. Maarof, and S.~Z.~M. Shaid, ``Ransomware threat success
  factors, taxonomy, and countermeasures: A survey and research directions,''
  \emph{Computers \& Security}, vol.~74, pp. 144--166, 2018.

\bibitem{8685181}
S.~D. {S.L} and J.~{CD}, ``Windows malware detector using convolutional neural
  network based on visualization images,'' \emph{IEEE Transactions on Emerging
  Topics in Computing}, pp. 1--1, 2019.

\bibitem{ALAZAB201591}
\BIBentryALTinterwordspacing
M.~Alazab, ``Profiling and classifying the behavior of malicious codes,''
  \emph{Journal of Systems and Software}, vol. 100, pp. 91 -- 102, 2015.
  [Online]. Available:
  \url{http://www.sciencedirect.com/science/article/pii/S0164121214002283}
\BIBentrySTDinterwordspacing

\bibitem{demme2013feasibility}
J.~Demme, M.~Maycock, J.~Schmitz, A.~Tang, A.~Waksman, S.~Sethumadhavan, and
  S.~Stolfo, ``On the feasibility of online malware detection with performance
  counters,'' \emph{ACM SIGARCH Computer Architecture News}, vol.~41, no.~3,
  pp. 559--570, 2013.

\bibitem{burgcsvfiles}
\BIBentryALTinterwordspacing
G.~J.~J. van~den Burg, A.~Naz{\'{a}}bal, and C.~Sutton, ``Wrangling messy {CSV}
  files by detecting row and type patterns,'' \emph{CoRR}, vol. abs/1811.11242,
  2018. [Online]. Available: \url{http://arxiv.org/abs/1811.11242}
\BIBentrySTDinterwordspacing

\bibitem{fisher2008dirt}
K.~Fisher, D.~Walker, K.~Q. Zhu, and P.~White, ``From dirt to shovels: fully
  automatic tool generation from ad hoc data,'' \emph{Acm sigplan notices},
  vol.~43, no.~1, pp. 421--434, 2008.

\bibitem{rowe2011finding}
N.~C. Rowe and S.~L. Garfinkel, ``Finding anomalous and suspicious files from
  directory metadata on a large corpus,'' in \emph{International Conference on
  Digital Forensics and Cyber Crime}.\hskip 1em plus 0.5em minus 0.4em\relax
  Springer, 2011, pp. 115--130.

\bibitem{Scofield_2017}
D.~Scofield, C.~Miles, and S.~Kuhn, ``Fast model learning for the detection of
  malicious digital documents,'' in \emph{SSPREW-7}, December 2017.

\bibitem{lundberg2005classifying}
J.~Lundberg, ``Classifying dialects using cluster analysis,'' \emph{Master's
  thesis, G{\"o}teborg University}, 2005.

\bibitem{grieve2011statistical}
J.~Grieve, D.~Speelman, and D.~Geeraerts, ``A statistical method for the
  identification and aggregation of regional linguistic variation,''
  \emph{Language Variation and Change}, vol.~23, no.~2, pp. 193--221, 2011.

\bibitem{yong2013beginner}
A.~G. Yong, S.~Pearce \emph{et~al.}, ``A beginner’s guide to factor analysis:
  Focusing on exploratory factor analysis,'' \emph{Tutorials in quantitative
  methods for psychology}, vol.~9, no.~2, pp. 79--94, 2013.

\bibitem{common_crawl}
C.~C. Foundation, ``Common {Crawl},'' \url{http://commoncrawl.org}, 2021,
  [Online; accessed 11-Mar-2021].

\bibitem{Robinson_2021}
M.~Robinson, ``Cosheaf representations of relations and dowker complexes,''
  \emph{J Appl. and Comput. Topology}, 2021.

\bibitem{arnold1991pseudolikelihood}
B.~C. Arnold and D.~Strauss, ``Pseudolikelihood estimation: some examples,''
  \emph{Sankhy{\=a}: The Indian Journal of Statistics, Series B}, pp. 233--243,
  1991.

\bibitem{plotly}
\BIBentryALTinterwordspacing
Plotly, ``3d network graphs in python/v3,'' 2022. [Online]. Available:
  \url{https://plotly.com/python/v3/3d-network-graph/}
\BIBentrySTDinterwordspacing

\bibitem{igraph}
\BIBentryALTinterwordspacing
igraph, ``python-igraph manual,'' 2020. [Online]. Available:
  \url{https://igraph.org/python/tutorial/latest/tutorial.html}
\BIBentrySTDinterwordspacing

\end{thebibliography}

\end{document}